\documentclass[11pt]{article}
\usepackage{authblk}
\usepackage{hyperref}
\usepackage{graphicx}

\usepackage{latexsym}
\usepackage{amsfonts, amsmath}
\usepackage{amssymb}
\usepackage{theorem}
\usepackage{longtable}
\usepackage{soul}
\usepackage[usenames]{color}
\usepackage{times}
\usepackage{tikz}

\usepackage[normalem]{ulem}
\usepackage{cancel}

\usepackage{xcolor}
\usetikzlibrary{matrix, positioning}

\usepackage{comment}
\usepackage{pgfplots}
\pgfplotsset{compat=1.15} 
\usetikzlibrary{patterns}
\usepackage{enumitem}
\usepackage{bbm}
\usepackage{multicol}
\usepackage{float}
\usepackage{setspace}
\usepackage[natbibapa]{apacite}

\usepackage[margin=2cm]{geometry}

\newcommand{\SVD}{\texttt{SVD}}

\newtheorem{theorem}{Theorem}[section]
\newtheorem{definition}[theorem]{Definition}
\newtheorem{remark}[theorem]{Remark}
\newtheorem{proposition}[theorem]{Proposition}
\newtheorem{corollary}[theorem]{Corollary}
\newtheorem{lemma}[theorem]{Lemma}

\newtheorem{example}[theorem]{Example}

\def\ND{\succsim\kern-11pt/\kern5pt}

\newenvironment{proof}
{\begin{trivlist}\item[]{Proof:}}{\hfill{$\square$}\noindent\end{trivlist}}

\title{Principal Component Analysis and Power Indices}

\author[1,2]{Xavier Molinero}
\author[1,2]{Enric Monsó}
\author[1,3]{Daniel Samaniego}

\affil[1]{Department of Mathematics (MAT)\\
Universitat Politècnica de Catalunya - BarcelonaTech (UPC)}

\affil[2]{Campus Terrassa, Edif. TR5, Colom 11, E-08222 Terrassa, Spain}
\affil[3]{Campus Manresa, Edif. MN3, Bases de Manresa, 61-73, E-08242 Manresa, Spain}

\date{}

\begin{document}
\maketitle

\begin{abstract}
Measuring the influence of a player in a simple game is a widely studied topic. Shapley-Shubik power index is perhaps the maximum exponent in terms of relevance. Furthermore, other power indexes have been proposed over time. In this paper, we propose yet another power index, defined in terms of winning coalitions. We show that this index coincides with eigenvalues obtained with the Principal Component Analysis method, a broadly used technique in data science to determine the influence of different features given a dataset. Furthermore, we provide a characterization of this proposed index in terms of four properties.  
 \medskip
  
  
\end{abstract}

\section{Introduction and Preliminaries}
The notion of measuring power is widely applied in many real world situations as well as broadly studied in the literature~\citep{Taylor1996,Keijzer2008,Bolus2011}.
In recent years, many applications in Game Theory have been used in the machine learning context and vice-versa ~\citep{Odu2024,McBride2026, SperryTaylor2027}. The main approach links these two fields considering the features in a machine learning problem as players in a game, as well as the target feature of the problem as the value of the characteristic function of the game.

In this article, we present a power index for simple games. This power index assigns to each player the corresponding contribution achieved by the Singular Vector Decomposition method (Theorem \ref{thm:svd}) when representing the game through a matrix in a certain way. As shown in the Theorem, the computation of the index is easily obtained in terms of the winning coalitions of the game, while the Singular Value Decomposition in other contexts can derive into complex expressions in terms of the initial information (see \cite{Ri89}) or directly into numerical solutions (see \cite{GoKa65}). 

A frequent situation in machine learning is related to measure the influence of different features given a dataset. One approach to do so is the Principal Component Analysis, PCA. This approach  was introduced from a geometrical point of view by \cite{Pe1901}. The technique involves Singular Value Decomposition, SVD, and uses its eigenvalues to measure the influence of the stated features \citep{Ja91,Joll02,joll16} among others). Furthermore, another circumstance where SVD is relevant is when it is of interest to measure the features' influence \emph{over a target} \citep{ba06,bar11} among others). 
These characterizations enrich the literature focused on characterizing power indices using a transfer property (\citep{Du75,DuSh79,Freixas2011,FrSa23}.

\subsection{Notions on simple games}

Let $N=\{1,2,3,\dots,n\}$ be the set of \textit{players}. Also denote $[n]=\{1,\dots,n\}$. Every subset of players $ S\subseteq N$ is referred as a \textit{coalition} and can be noted as an $n$-vector $\vec{s}\in\{0,1\}^n$ whereas $s_i=1$ when player $i$ belongs to the coalition $S$ and $s_i=0$ when player $i$ does not belong to $S$(where $s_i=1$ indicates that player $i$ belongs to the coalition $S$ while $s_i=0$ indicates that play $i$ does not belong to $S$) . As it is customary in the literature $2^N=\{S \ | \ S\subseteq N\}$ is the set of all coalitions. A \textit{simple game} is a pair $\Gamma=(N,v)$ where 
\begin{align*}
  v \colon 2^N &\to \{0,1\}\\
  S &\mapsto v(S)
\end{align*}
indicates whether a coalition wins the game (i.e. $v(S)=1$) or does not (i.e. $v(S)=0$), and satisfies the following properties: (en una enumeration millor?)
\begin{enumerate}
    \item $v(\emptyset)=0,$ (void coalition has to lose the game);
    \item $v(N)=1,$ (grand coalition wins the game for sure);
    \item $\forall{}S\subseteq{}T\subseteq{}N$ then $v(S)\leq v(T)$,  the \textit{monotonicity} property (incorporation of players do not diminish chances to win the game).
\end{enumerate}

For a comprehensive introduction to simple games, we refer to \cite{taylor_zwicker_book}.

Given a simple game $\Gamma=(N,v)$, let $\mathcal{W}(\Gamma)$ be the set of all winning coalitions. It is customary to define a game $\Gamma$ by listing the coalitions in $\mathcal{W}$. For $i\in[n]$, let $\mathcal{W}_i(\Gamma)$ be the set of all winning coalitions containing player $i$, that is,
$\mathcal{W}_i(\Gamma)=\{S\in\mathcal{W} \ | \ i\in{}S\}$, 
and
$\mathcal{W}^i(\Gamma)$ denotes the set of winning coalitions of size $i$, 
$\mathcal{W}^i(\Gamma)=\{S \in \mathcal{W}(\Gamma) \mid |S| = i\}$.
Clearly, $\mathcal{W}(\Gamma)=\bigcup_{i\in[n]} \mathcal{W}^i(\Gamma)$, and these sets are pairwise disjoint.
A winning coalition which strictly does not contain another winning coalition is called \textit{minimal winning coalition}. Finally, $\mathcal{W}^m(\Gamma)=\{S\in\mathcal{W} \ | \ v(T)=0,\forall{}T\subsetneq S\}$ denotes the set of all minimal winning coalitions.

If there is no ambiguity, we simply use the notation $\mathcal{X}$ to refer to $\mathcal{X}(\Gamma)$, where
$\mathcal{X} \in \{\mathcal{W}, {\cal W}_i, {\cal W}^i, {\cal W}^m\}$, for $i\in[n]$.

To classify the players of the game according to its influence in the game, \cite{Is58} introduced the following relation.  Let $\tau_{ij}:N\rightarrow N$ be the transposition of players $i,j\in N$. We say that $i$ is \emph{at least as desirable as} $j$, $i\succsim j$, if and only if $\tau_{ij}(\mathcal{W}_j)\subseteq\mathcal{W}_i$.
We say that both players are \emph{equally desirable} when $i\succsim j$ and $j\succsim i$ and we denote it by $i\simeq j$.
In other words, replacing player $j$ in any winning coalition by player $i$, maintains the coalition winning. The relation $\succsim$ is a preorder.

A \emph{null player} $i$ is a player whose decision does not affect the outcome of the game, formally $v(S\cup \{i\})=v(S)$ for any coalition $S\subseteq N\backslash\{i\}$. A \emph{dummy player} $i$ satisfies  $v(S\cup \{i\})-v(S)=v(\{i\})$ for any coalition $S\subseteq N\backslash\{i\}$. We say a player $i$ is a \emph{blocker} (or a \emph{vetoer}) if it belongs to any winning coalition, that is $i\in S$ for all $S\in\mathcal{W}$.

As usual, given two simple games $\Gamma_1=(N,v)$ and $\Gamma_2=(N,w)$, we define the union and the intersection as $\Gamma_1\cup\Gamma_2=(N,v\cap w)$ and $\Gamma_1\cup\Gamma_2=(N,v\cup w)$, respectively, where $\forall{}S\subseteq{}N$
\begin{equation*}
    (v \cap w)(S)=\max(v(S),w(S))\  \ \ \ \text{ and }\  \ \ \ (v \cup w)(S)=\min(v(S),w(S))
\end{equation*}
Moreover, if $\pi$ is a permutation of $N$, we define $\pi{}v(S)=v(\pi^{-1}(S))$.

Let $\Gamma=(N,v)$ be a simple game. We define the minimum size of a winning coalition as $k(\Gamma)=\min\limits_{S\in \mathcal{W}} |S|$, and the number of winning coalitions of cardinality $k(\Gamma)$ by $p(\Gamma)=|\{S\in\mathcal{W} \ | \ |S|=k(\Gamma)\}|$.

\begin{definition}\label{def:matrixrepresentation}
Given a simple game $\Gamma=(N,v)$, we define $\mathsf{V}(\Gamma)=[\mathsf{S}\,|v(\mathsf{S})] \in {\cal M}^{2^n\times(n+1)}$ as its matrix representation in the following way.
Each row corresponds to a coalition $S\subseteq{}N$, where the first $n$ columns coincide with the vector $\vec{s}\in\{0,1\}^n$ associated with $S$ and the value of the last column is $v(S)$.
Moreover, by convention, we order the rows in lexicographical order. We denote by $\mathcal{V}^{2^n\times(n+1)}$ the set of all the matrices of this kind.
\end{definition}

\begin{example}\label{ex:4pl}
    Let $\Gamma=(N,v)$ be a simple game defined by $N=\{1,2,3,4\}$, and $v$ indicating that coalitions in 
    $\mathcal{W}=\{\{1\},\{1,2\},\{1,3\},\{1,4\},\{2,3\},\{1,2,3\},\{1,2,4\},\{1,3,4\},\{2,3,4\},\{1,2,3,4\}\}$ are the winning coalitions.
    Hence, the corresponding matrix representation is
    \[\textsf{V}(\Gamma)=
\left(
{\scriptsize
\begin{array}{cccc|c}
0&0&0&0&0\\
0&0&0&1&0\\
0&0&1&0&0\\
0&0&1&1&0\\
0&1&0&0&0\\
0&1&0&1&0\\
0&1&1&0&1\\
0&1&1&1&1\\
1&0&0&0&1\\
1&0&0&1&1\\
1&0&1&0&1\\
1&0&1&1&1\\
1&1&0&0&1\\
1&1&0&1&1\\
1&1&1&0&1\\
1&1&1&1&1\\
\end{array}
}
\right)\in\mathcal{V}^{2^4\times(4+1)}\ .
\]
\end{example}

\begin{remark}
\label{remark:1-0}
For each $j\in[n]$, the $j$-th column has mean and standard deviation equal to $1/2$, denoted by 
$\mu_{\bullet j}(\textsf{V}(\Gamma))$ and $\sigma_{\bullet j}(\textsf{V}(\Gamma))$, respectively. 
Furthermore, each of these columns contains exactly $2^{n-1}$ ones and $2^{n-1}$ zeros. 
On the other hand, the last column contains $|\mathcal{W}(\Gamma)|$ ones, corresponding to the winning coalitions, and 
$2^{N} - |\mathcal{W}(\Gamma)|$ zeros, corresponding to the non-winning coalitions.
\end{remark}

\begin{remark}
Two different simple games will have different matrix representations. However, not any matrix ${\cal M}^{2^n\times(n+1)}$ is the representation of a simple game. For example, 
\[
\left(
{\scriptsize
\begin{array}{ccc|c}
0&0&0&0\\
0&0&1&0\\
0&1&0&0\\
0&1&1&1\\
{\bf 1}&{\bf 0}&{\bf 0}&{\bf 1}\\
1&0&1&1\\
{\bf 1}&{\bf 1}&{\bf 0}&{\bf 0}\\
1&1&1&1
\end{array}
}
\right)\notin\mathcal{V}^{2^3\times(3+1)}\ .
\]
does not correspond to a simple game as it does not respect monotonicity. For instance, for this game, coalition $\{1\}$ is winning, while $\{1,2\}$ is non-winning.

Determining how many of the $2^{2^n}$ possible choices for the matrix last column come from a simple game is equivalent to count how many Boolean functions are there. This is an open problem only solved up to $n\leq 9$ (see, for instance, \cite{JA23} and \cite{Hi24} for the last known results, and \href{https://oeis.org/A000372}{A000372} at the On-Line Encyclopedia of Integer Sequences (OEIS) for the whole known sequence).
\end{remark}

Prior to defining the orthogonal (normalized) matrix of a given simple game in accordance with the standardization criteria of~\cite{Kut2005}, we state the following.

\begin{remark}
\label{remark:wij}
Given a simple game $\Gamma=(N,v)$ and its matrix representation $\mathsf{V}(\Gamma)=[v_{ij}]\in{\cal V}^{2^n\times(n+1)}$, for $i\in[2^n]$ and $j\in[n]$, we have
\[
\widehat{w}_{ij} :=
\frac{v_{ij} - \mu_{\bullet j}(\textsf{V}(\Gamma))}
     {\sigma_{\bullet j}(\textsf{V}(\Gamma))}
=
\begin{cases}
  \;\;1, & \text{if } v_{ij} = 1,\\[4pt]
 -1, & \text{if } v_{ij} = 0.
\end{cases}
\]
Since 
$
|\{S \subseteq N \mid i,j \in S\}| 
= |\{S \subseteq N \mid i \in S,\, j \notin S\}|
 = |\{S \subseteq N \mid i \notin S,\, j \in S\}|
= |\{S \subseteq N \mid i,j \notin S\}| 
= 2^{n-2}$,
the number of $1$s and $-1$s in
$\widehat{w}_{\bullet i}\cdot\widehat{w}_{\bullet j}$
is the same for all $i,j \in [n]$ with $i \neq j$. Hence,
$\widehat{w}_{\bullet i} \cdot \widehat{w}_{\bullet j} = 0$.

Furthermore, by Remark~\ref{remark:1-0}, it follows that
$\widehat{w}_{\bullet j} \cdot \widehat{w}_{\bullet j} = 2^n$, for all $j \in [n]$.
\end{remark}

\begin{definition}
\label{def:normalrepresentation}
Given a simple game $\Gamma=(N,v)$ and its matrix representation $\mathsf{V}(\Gamma)=[v_{ij}]\in{\cal V}^{2^n\times(n+1)}$, we define the orthogonal (normalized) matrix of $\Gamma$ as $\mathsf{W}(\Gamma)=[w_{ij}]\in{\cal M}^{2^n\times(n+1)}$ entrywise as follows:
\begin{equation*}
  w_{ij} = 
  \begin{cases} 
    \alpha_{\bullet{}j}\cdot
    \dfrac{v_{ij}-\mu_{\bullet{}j}(\textsf{V}(\Gamma))}{\sigma_{\bullet{}j}(\textsf{V}(\Gamma))} & \text{ if }j\neq{}n+1\\
     v_{ij} & \text{ if }j=n+1 
  \end{cases}\ ,
\end{equation*}
where
$\mu_{\bullet{}j}(\textsf{V}(\Gamma))$ is the mean of the $j$-th column of $\textsf{V}(\Gamma)$,
$\sigma_{\bullet{}j}(\textsf{V}(\Gamma))$ is the standard deviation of the $j$-th column of $\textsf{V}(\Gamma)$,
and $\alpha_{\bullet{}j}=\dfrac{1}{\sqrt{2^n}}$ so that the first $n$ columns of $\textsf{W}(\Gamma)$ are orthogonal and orthonormal.

Furthermore, as $\mu_{\bullet{}j}(\textsf{V}(\Gamma))=\sigma_{\bullet{}j}(\textsf{V}(\Gamma))=\frac12$,
we obtain
\begin{equation*}
  w_{ij} = 
  \begin{cases}
     \dfrac{2\,v_{ij}-1}{\sqrt{2^n}} & \text{ if }j\neq{}n+1\\
     v_{ij} & \text{ if }j=n+1 
  \end{cases}\ .
\end{equation*}
\end{definition}

We now apply this definition to the previous example.
\begin{example}\label{ex:4pl2}
    Let $\Gamma=(N,v)$ be the simple game defined in Example~\ref{ex:4pl}. From its matrix representation $\mathsf{V}(\Gamma)\in{\cal V}^{2^n\times(n+1)}$, the associated orthogonal matrix turns to be
    \[\textsf{W}(\Gamma)=
\left(
{\scriptsize
\begin{array}{rrrr|c}
-1/\sqrt{2^4}&-1/\sqrt{2^4}&-1/\sqrt{2^4}&-1/\sqrt{2^4}&0\\
-1/\sqrt{2^4}&-1/\sqrt{2^4}&-1/\sqrt{2^4}&1/\sqrt{2^4}&0\\
-1/\sqrt{2^4}&-1/\sqrt{2^4}&1/\sqrt{2^4}&-1/\sqrt{2^4}&0\\
-1/\sqrt{2^4}&-1/\sqrt{2^4}&1/\sqrt{2^4}&1/\sqrt{2^4}&0\\
-1/\sqrt{2^4}&1/\sqrt{2^4}&-1/\sqrt{2^4}&-1/\sqrt{2^4}&0\\
-1/\sqrt{2^4}&1/\sqrt{2^4}&-1/\sqrt{2^4}&1/\sqrt{2^4}&0\\
-1/\sqrt{2^4}&1/\sqrt{2^4}&1/\sqrt{2^4}&-1/\sqrt{2^4}&1\\
-1/\sqrt{2^4}&1/\sqrt{2^4}&1/\sqrt{2^4}&1/\sqrt{2^4}&1\\
1/\sqrt{2^4}&-1/\sqrt{2^4}&-1/\sqrt{2^4}&-1/\sqrt{2^4}&1\\
1/\sqrt{2^4}&-1/\sqrt{2^4}&-1/\sqrt{2^4}&1/\sqrt{2^4}&1\\
1/\sqrt{2^4}&-1/\sqrt{2^4}&1/\sqrt{2^4}&-1/\sqrt{2^4}&1\\
1/\sqrt{2^4}&-1/\sqrt{2^4}&1/\sqrt{2^4}&1/\sqrt{2^4}&1\\
1/\sqrt{2^4}&1/\sqrt{2^4}&-1/\sqrt{2^4}&-1/\sqrt{2^4}&1\\
1/\sqrt{2^4}&1/\sqrt{2^4}&-1/\sqrt{2^4}&1/\sqrt{2^4}&1\\
1/\sqrt{2^4}&1/\sqrt{2^4}&1/\sqrt{2^4}&-1/\sqrt{2^4}&1\\
1/\sqrt{2^4}&1/\sqrt{2^4}&1/\sqrt{2^4}&1/\sqrt{2^4}&1\\
\end{array}
}
\right).
\]
\end{example}

\subsection{Singular Value Decomposition and (supervised) Principal Component Analysis}

In this section, we review principal component analysis (PCA), supervised PCA, and the role of singular value decomposition (SVD) in these methods.

It is a well known result in linear algebra~\cite{Jordan1874,HoJo12} that every real matrix $\mathsf{M}\in\mathcal{M}^{m\times n}$ admits a factorization, so called singular value decomposition (SVD), such that
$$\mathsf{M}=\mathsf{U}\Sigma \mathsf{V}^T\,,$$
where 
$\mathsf{U}\in{\cal M}^{m\times m}$ and $\mathsf{V}\in {\cal M}^{n\times n}$ are orthogonal matrices, and
$\Sigma\in {\cal M}^{m\times n}$ has non-negative diagonal values and zero off-diagonal elements. These diagonal values $\sigma_1,\ldots,\sigma_n$ are the singular values of $\mathsf{M}$.
Moreover, 
the columns of $\mathsf{U}\in{\cal M}^{m\times m}$ form an orthonormal basis of $\mathbb{R}^m$, and the columns of $\mathsf{V}\in {\cal M}^{n\times n}$ give an orthonormal basis of $\mathbb{R}^n$. See Figure~\ref{fig:SVD}.

\begin{figure}[H]
\begin{center}
\begin{tikzpicture}[scale=0.9]

\def\h{2.4}      
\def\hV{1.2}      
\def\s{0.6}      

\def\wM{1.8}
\def\wU{2.4}
\def\wS{1.8}
\def\wV{1.8}

\draw[fill=gray!25] (0,0) rectangle (\wM,\h);
\foreach \x in {0,\s,1.2} \draw (\x,0) -- (\x,\h);
\foreach \y in {0,\s,1.2,1.8} \draw (0,\y) -- (\wM,\y);
\node at (0.9,-0.4) {\small $M\in{\cal M}^{m\times n}$};

\node at (2.3,1.2) {$=$};

\begin{scope}[shift={(2.8,0)}]
\draw[fill=teal!30] (0,0) rectangle (\wU,\h);

\foreach \x/\c in {
  0.0/orange!60,
  0.6/orange!30,
  1.2/yellow!60,
  1.8/white!35
}{
  \draw[fill=\c] (\x,0) rectangle (\x+0.6,\h);
}

\foreach \x in {0,0.6,1.2,1.8,2.4} \draw (\x,0) -- (\x,\h);
\foreach \y in {0,0.6,1.2,1.8,2.4} \draw (0,\y) -- (\wU,\y);

\node at (1.2,-0.4) {\small $\mathsf{U}\in{\cal M}^{m\times m}$};
\end{scope}

\begin{scope}[shift={(5.7,0)}]
\draw[fill=white] (0,0) rectangle (\wS,\h);

\foreach \x in {0,0.6,1.2,1.8} \draw (\x,0) -- (\x,\h);
\foreach \y in {0,0.6,1.2,1.8,2.4} \draw (0,\y) -- (\wS,\y);

\draw[fill=orange!60] (0,1.8) rectangle (0.6,2.4) node[midway] {\footnotesize $\sigma_1$};
\draw[fill=orange!30] (0.6,1.2) rectangle (1.2,1.8) node[midway] {$\ddots$};
\draw[fill=yellow!60] (1.2,0.6) rectangle (1.8,1.2) node[midway] {\footnotesize $\sigma_n$};

\node at (0.9,-0.4) {\small $\Sigma\in{\cal M}^{m\times n}$};
\end{scope}

\begin{scope}[shift={(8.0,0)}]
\draw[fill=purple!30] (0,0) rectangle (\wV,\hV);

\foreach \y/\c in {
  1.2/orange!60,
  0.6/orange!30,
  0/yellow!60
}{
  \draw[fill=\c] (0,\y) rectangle (\wV,\y+0.6);
}

\foreach \x in {0,0.6,1.2} \draw (\x,0) -- (\x,\hV+.6);
\foreach \y in {0,0.6,1.2} \draw (0,\y) -- (\wV,\y);

\node at (.9,-0.4) {\small $\mathsf{V}^T\in{\cal M}^{n\times n}$};
\end{scope}

\end{tikzpicture}
\end{center}
\caption{Matrix representation of the SVD. W.l.o.g., we assume $m\geq n$.}
\label{fig:SVD}
\end{figure}
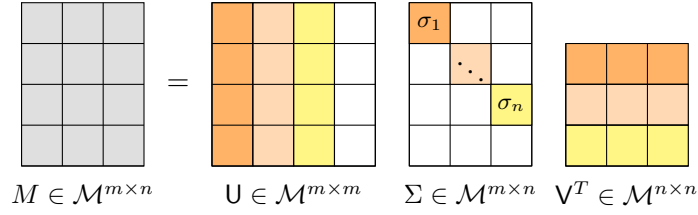

Since the singular value decomposition uniquely determines the matrix $\Sigma$, but not the associated matrices $\mathsf{U}$ and $\mathsf{V}$ (for more detail, see Theorem 2.6.3 in~\cite{HoJo12}), we can define the mapping that gets these singular values
\[
\begin{aligned}
\SVD{} :\  \mathcal{M}^{m\times n} &\longrightarrow \mathbb{R}^q \\
   \mathsf{M} &\longmapsto (\sigma_1,\dots,\sigma_q)
\end{aligned}
\]
where $q=\min\{m,n\}$.

Note that, in the literature, eigenvalues are typically arranged in decreasing order for the purposes of component analysis~\cite{HoJo12,Meyer2023}. However, in this work we preserve their original ordering, as is common in statistics. Given a data matrix, principal component analysis (PCA)~\cite{Pe01,Joll02} is a well-known dimension-reduction method that relies on its singular value decomposition (SVD) to assess the relevance of the components. This procedure can also reduce computational costs in data analysis by discarding non-essential information.

Furthermore, the PCA procedure can be extended to include the analysis of principal components with respect to a target, name it $y\in\mathbb{R}^m$, (\cite{ba06}, \cite{bar11}) in the so called supervised PCA. Given a data set matrix $\mathsf{X}\in{\cal M}^{m\times n}$ and a target vector $y\in\mathbb{R}^m$, (or dataset with target $\mathsf{X}|y\in\mathcal{M}^{m\times(n+1)}$), the following mapping is defined:

\[
\begin{aligned}
\tau:\  \mathcal{M}^{m\times (n+1)} &\longrightarrow  \mathcal{M}^{m\times n} \\
   \mathsf{X}|y &\longmapsto \mathsf{X}^*
\end{aligned}
\]

where, we consider the columns of $\mathsf{X}$ centred at $0$ removing its column's mean, as $x^{cent}_{i,j}=x_{i,j}-\overline{x}_j$, and

\begin{equation}\label{eq:st_coef}
    x^*_{i,j}=x^{cent}_{i,j}\frac{(x^{cent}_i)^T\cdot y}{\sqrt{(x^{cent}_i)^T\cdot x^{cent}_i}},\quad\forall i\in[m]\text{ and }j\in[n],\,
\end{equation}
with $\overline{x}_j$ denoting the average of the elements of the column $j$. Also recall that, given two vectors $u,v\in\mathbb{R}^n$, we define the real number $u\cdot v=\sum_{j=1}^n u_jv_j$. Note that $\SVD \ \circ \tau$ is a well defined mapping and, for simplicity, we set $\SVD_\tau=\SVD \ \circ \tau$ whose domain is $\mathcal{M}^{m\times(n+1)}$ and its image is $\mathbb{R}^{\min\{n,m\}}$.

The next example illustrates this procedure and has been computed using the \emph{Python} library 
\emph{SciPy}, its submodule \emph{linalg} calling the function \emph{svd} (see \cite{py16}).

\begin{example}
Let $X|y\in\mathcal{M}^{4\times 4}$ be the dataset with target defined  as

$$X|y=
\left(
\begin{array}{rrr|c}
3 & 1 & 2 & 6 \\
4 & 2 & 5 & 7 \\
5 & 1 & 2 & \sqrt{90} \\
6 & 0 & 3 &  11 \\
\end{array}
\right)
$$

We compute 

\[\tau(X|y)=
\begin{pmatrix}
\frac{-3}{2}\frac{8+3\sqrt{10}}{10} & 0 & -1\frac{8-3\sqrt{10}}{5} \\
\frac{-1}{2}\frac{8+3\sqrt{10}}{10} & \frac{-4}{5} & 2\frac{8-3\sqrt{10}}{5} \\
\frac{1}{2}\frac{8+3\sqrt{10}}{10} & 0 & -1\frac{8-3\sqrt{10}}{5} \\
\frac{3}{2}\frac{8+3\sqrt{10}}{10} & \frac{4}{5} & 0
\end{pmatrix}
\]
and then we get 
\[\SVD(\tau(X|y))=
(3.978 ,1.056 ,0.396).
\]
\end{example}

Nota that for a general dataset with target $X|y\in\mathcal{M}^{m\times(n+1)}$, the singular values are related to the components of the singular value decomposition and not the initial columns itself. In Section \ref{S:4}, the procedure will be interpreted for matrices $\mathsf{V}(\Gamma)\in\mathcal{M}^{2^n\times (n+1)}$ where $\Gamma$ is a simple game with $n$ players.

\section{The \SVD{} power index}
The next definition states the proposed power index that will be studied in the main results of the article. 

\begin{definition} $\SVD{}$ power index. Let $\Gamma = (N,v)$ be a simple game. The \SVD{} power index for game $\Gamma$ is defined, for each player $i \in N$, as
\[
\SVD_i
= |\mathcal{W}_i|-\frac{|\mathcal{W}|}{2}.
\]
Consequently the normalized $\SVD$ power index for game $\Gamma$ and player $i$ is the corresponding ratio
\[
\overline{\SVD}_i
= \frac{\SVD_i}{\sum_{j \in N} \SVD_j}.
\]
\end{definition}

We shall denote $\SVD_i(\Gamma)$ and $\overline{\SVD}_i(\Gamma)$ by $\SVD_i$ and $\overline{\SVD}_i$ whenever ambiguity arises.

\begin{example} Let $(N,v)$ be the simple game defined at Example~\ref{ex:4pl} by stating that the set of ten winning coalitions
$$\mathcal{W}=\{\{1\},\{1,2\},\{1,3\},\{1,4\},\{2,3\},\{1,2,3\},\{1,2,4\},\{1,3,4\},\{2,3,4\},\{1,2,3,4\}\}$$
Hence, for every player $i\in [4]$
\begin{align*}
       \mathcal{W}_1&=\{\{1\},\{1,2\},\{1,3\},\{1,4\},\{1,2,3\},\{1,2,4\},\{1,3,4\},\{1,2,3,4\}\}, \\
       \mathcal{W}_2&=\{\{1,2\},\{2,3\},\{1,2,3\},\{1,2,4\},\{2,3,4\},\{1,2,3,4\}\}, \\
       \mathcal{W}_3&=\{\{1,3\},\{2,3\},\{1,2,3\},\{1,3,4\},\{2,3,4\},\{1,2,3,4\}\}, \\
       \mathcal{W}_4&=\{\{1,4\},\{1,2,4\},\{1,3,4\},\{2,3,4\},\{1,2,3,4\}\},
    \end{align*}
and $|\mathcal{W}_1|=8$, $|\mathcal{W}_2|=|\mathcal{W}_3|=6$, $|\mathcal{W}_4|=5$. Therefore, being $|\mathcal{W}|=10,$ we have that $\SVD(1)=3, \SVD(1)=2,\SVD(3)=1,$ and $\SVD(4)=0$ or equivalently
$\SVD=\left(3,1,1,0\right)$.

Moreover, the corresponding normalized \SVD~power index turns to be $\overline{\SVD}=\left(\frac{3}{5},\frac{1}{5},\frac{1}{5},0\right)$.

We notice that the Shapley-Shubik power index for this simple game is $\left(\frac{2}{3},\frac{1}{6},\frac{1}{6},0\right)$. And remark the similar classification of players that both index obtain: $1$ is the most powerful player, $2$ and $3$ share the same amount of influence and $4$ is a dumb player seen from the two indices points of view.

\end{example}

Now, we present an example with $5$ players.
\begin{example}
\label{ex:5pl}
Let $(N,v)$ be a simple game defined $N=[5]$ and $\mathcal{W}^m=\{\{1,2\},\{2,4\},\{4,5\},\{2,3,5\}\}$.
Then $|\mathcal{W}|=17$, $|\mathcal{W}_1|=10$, $|\mathcal{W}_2|=13$,
$|\mathcal{W}_3|=9$,
$|\mathcal{W}_4|=12$,
and $|\mathcal{W}_5|=11$.
Hence,
$\SVD=\left(\frac32,\frac92,\frac12,\frac72,\frac52\right)$ and
$\overline{\SVD}=\left(\frac3{25},\frac9{25},\frac1{25},\frac7{25},\frac5{25}\right)$. In this case, \SVD~also coincides with the Banzhaf power index.

Additionally, Shapley-Shubik power index is given by $\left(\frac{7}{60},\frac{11}{30},\frac{1}{30},\frac{17}{60},\frac{1}{5}\right)$. Once again, as expected, the same ranking is obtained, now with only slight variations in the values.
\end{example}

The following example gives us the same value for \SVD, Shaple-Shubik and Banzhaf power indices.
\begin{example}
Let $(N,v)$ be a simple game defined $N=[4]$ and $\mathcal{W}^m=\{\{1,2\},\{1,3\},\{2,3,4\}\}$.
Now, it results
$|\mathcal{W}|=7$, $|\mathcal{W}_1|=6$, $|\mathcal{W}_2|=5$,
$|\mathcal{W}_3|=5$
and $|\mathcal{W}_4|=4$.
Hence,
$\overline{\SVD}=(\frac5{12},\frac14,\frac14,\frac1{12})$, as well as Shapley-Shubik and Banzhaf.
\end{example}

In the previous two examples both indexes, SVD and Shapley-Shubik, are not in contradiction regarding the strict order of importance of the viewpoint players.
However, for Example~\ref{ex:5pl}, we have $SVD_2\geq SVD_5 = SVD_4\geq SVD_1=SVD_3 $,
while for Shapley-Shubik it is ${SS}_2>{SS}_5>{SS}_4>{SS}_1>{SS}_3$.

\begin{proposition}
\label{prop:null}
Given a simple game $(N,v)$, for each $i\in{}N$, the SVD power index satisfies $SVD_i(\Gamma)\ge0$.
Specifically, $SVD_i(\Gamma)=0$ if $i$ is a dummy player, and $SVD_i(\Gamma)>0$ otherwise.
\end{proposition}
\begin{proof}
\begin{enumerate}[label=(\alph*)]
\item If $i$ is a null player, for any $S\subseteq{}N\backslash\{i\}$, clearly $S\cup\{i\}\in \mathcal{W}$ if and only if $S\in\mathcal{W}$. Hence exactly half of the elements in $\mathcal{W}$ are in $\mathcal{W}_i$, i.e., $|\mathcal{W}_i|-\frac{|\mathcal{W}|}{2}=0$.

\item If $i$ is not a null player, consider the set $\mathcal{W}_{-i}=\{S\in\mathcal{W} \ | \ i\notin S\}$.
It is clear that $\mathcal{W}_{i}\cup\mathcal{W}_{-i}=\mathcal{W}$. By monotonicity,
$\forall{}S\in\mathcal{W}_{-i}$, then
$S\cup\{i\}\in\mathcal{W}_{i}$.
Thus
$|\mathcal{W}_{-i}|\le|\mathcal{W}_i|$.
As $i$ is not null player: $\exists{}S\in{\cal W}_i$ such that $S\setminus\{i\}\notin{\cal W}$. Hence, it results that $|\mathcal{W}_i|>|\mathcal{W}_{-i}|\ge\frac{|\mathcal{W}|}{2}$, i.e.,
$SVD(i)=|\mathcal{W}_i|-\frac{|\mathcal{W}|}{2}>0$. \end{enumerate}\end{proof}
    
\begin{remark}
    If players $i,j$ are symmetric, then $|\mathcal{W}_i|=|\mathcal{W}_j|$, so $SVD\,_i(\Gamma)=SVD{}_j(\Gamma)$.
\end{remark}

\begin{remark}
    If players $i,j$ verify $i \succsim j$, then $|\mathcal{W}_i|\geq |\mathcal{W}_j|$, hence $SVD_i(\Gamma)\ge SVD_j(\Gamma)$.
\end{remark}

There are, in the literature, other compelling properties that an arbitrary power index $\phi$ should satisfy. For example, we state the following two and we prove SVD satisfies both.

\begin{definition}
     Given two games $\Gamma_1=(N_1,v_1)$ and $\Gamma_2=(N_2,v_2)$ with $N_1\cap N_2=\emptyset$, we define the intersection game by
     \begin{equation*}
         \Gamma_1\cap \Gamma_2 =(N_1\cup N_2,w)
     \end{equation*} 
     where $w(S)=1$ if and only if $v_1(S\cap N_1)=1$ and $v_2(S\cap N_2)=1$.
\end{definition}

\begin{definition}[Bicameral join postulate. Definition 2.4 in \cite{Fe98}]
    Given two games $\Gamma_1=(N_1,v_1)$ and $\Gamma_2=(N_2,v_2)$ with $N_1\cap N_2=\emptyset$, 
    we say that a power index $\phi$ satisfies the bicameral join postulate if, for any non-dummy players $i,j\in N_1$, its ratio of power is preserved when joining another disjoint game, that is,

    \begin{equation*}
        \frac{\phi_i(\Gamma_1)}{\phi_j(\Gamma_1)}=\frac{\phi_i(\Gamma_1\cap \Gamma_2)}{\phi_j(\Gamma_1\cap \Gamma_2 )}.
    \end{equation*}    
\end{definition}

\begin{proposition}\label{p:bicameral}
   The SVD power index satisfies the bicameral join postulate.

\end{proposition}

\begin{proof}
Note $\phi_i(\Gamma_1\cup \Gamma_2)=|\mathcal{W}_i(\Gamma_1)|\cdot |W(\Gamma_2)|-\frac{|\mathcal{W}(\Gamma_1)|\cdot |W(\Gamma_2)|}{2}$ and $\phi_j(\Gamma_1\cup \Gamma_2)=|\mathcal{W}_j(\Gamma_1)|\cdot |W(\Gamma_2)|-\frac{|\mathcal{W}(\Gamma_1)|\cdot |W(\Gamma_2)|}{2}$. The term $|W(\Gamma_2)|$ factors out and simplifies in the quotient, recovering the original ratio. Note that non-dumminess guaranties the denominator is non-zero. \end{proof}

\begin{definition}[Added Blocker postulate (see \cite{Fe16},\cite{FrSa23})]

We say that a power index $\phi$ satisfies the added blocker postulate if, for any simple game $\Gamma=(N,v)$, when adding a player $k$ that is a blocker in the resulting game, then, in the new game $\Gamma'$, the ratio of power of the original players is preserved,that is

    \begin{equation*}
        \frac{\phi_i(\Gamma_1)}{\phi_j(\Gamma_1)}=\frac{\phi_i(\Gamma')}{\phi_j(\Gamma' )}.
\end{equation*}

\end{definition}

\begin{corollary}
    The SVD power index satisfies the added blocker postulate
\end{corollary}

\begin{proof}
    Given a simple game $\Gamma=(N,v)$, when adding a player that is a blocker, is exactly to consider the game $\Gamma\cap (\{k\},w)$ where $w$ is defined only for $2^{|\{k\}|}=2$ values: $w(\emptyset)=0$ and $w(\{k\})=1$. Hence it is a particular case of Proposition \ref{p:bicameral} and the ratios are preserved.   
\end{proof}

The next properties are those to be used in the characterisation of SVD power index. 

\section{Relation between supervised PCA and the SVD index}\label{S:4}

Given a data set related to an experiment where $n$ features (players) are analyzed from $m$ points of view (coalitions), it is customary to consider a matrix $\mathsf{X}\in{\cal M}^{m\times n}$ where data is collected. Each of the $m$ rows represents a different repetition of the experiment, and each of the $n$ columns, $x_i\in\mathbb{R}^m$ for $i\in[n]$, contains data related to a particular feature.

From now on, we will consider datasets with target that are the matrix representation of a simple game $(N,v)$ stated in Definition \ref{def:matrixrepresentation}. Note that all these matrix representations are of the form of a dataset with target, hence, $\SVD_ {proc}$ can be applied. Hence, we consider the players as features, we state the dataset $X$ as the corresponding $\mathcal{P}(N)$ rows, and we set the target $y$ corresponding to the image of the characteristic function, $v$, of each coalition. 

\begin{lemma}\label{lemma:ortho}
     Given a simple game $\Gamma=(N,v)$, the columns of its matrix representation $\mathsf{X}^*(\Gamma)$ are orthogonal. Moreover, when multiplied by $\frac{1}{\sqrt{2^n}}$ they are orthonormal. Formally, for any $i,j\in N$,
     \begin{equation*}
         \frac{x_j^*}{\sqrt{2^n}}\cdot\frac{x_i^*}{\sqrt{2^n}}=\delta_{j,i}
     \end{equation*}
     where $\delta_{j,i}$ is the Kronecker delta.
 \end{lemma}
 \begin{proof}
     If $j=i$, it is clear that, component-wise $1\cdot 1=1$ and $-1\cdot -1=1$ with a total number of $2^n$ components. 
         If $j\neq i$, since the game is defined for all $S\in 2^N$, so we have 
     \begin{align*}
         |\{S\in 2^N \ | \ j,i\in S\}|&=|\{S\in 2^N \ | \ j\in S\text{ and }i\notin S\}|=|\{S\in 2^N \ | \ j\notin S\text{ and }i\in S\}|\\&=|\{S\in 2^N \ | \ j,i\notin S\}|=2^{n-2}
     \end{align*}
     thus, the amount of $1$ is the same than $-1$, hence the total sum is $0$.
 \end{proof}

(since they are already an orthonormal basis)
\begin{theorem}\label{thm:svd}
    For every single game $\Gamma=(N,v)$, $$\SVD(\Gamma)=\SVD_{\tau}(\mathsf{X}^*(\Gamma)).$$ 
\end{theorem}
\begin{remark}Note that Remark~\ref{remark:wij} addresses the issue in $\SVD_{\tau}$ that $\sigma_j$ is related to its ``$j$-th orthonormal component" instead of the ``$j$-th original column": due to orthonormality no change of basis is produced, thus the $j$-th component corresponds exactly to the $j$-th original column.
\end{remark}
\begin{proof} (draft)
    The key idea of the proof is that the matrix representing a simple game has a very manageable structure in this context, thus, most of the computations are simplified.
\begin{enumerate}

There is a wide range of feature selection models that can be classified in two main families: supervised and unsupervised. Supervised models directly exploit the class of the instances for determining if a feature is relevant or not, selecting those that are highly correlated with final target classes.

Considering normalized data, $y\leftarrow \dfrac{1}{\sigma_y}(y-\overline{y}),\,\, X_i\leftarrow \dfrac{1}{\sigma_{X_i}}(X_i-\overline{X_i}),$ assume that every player has a role in determining the game i.e. $y=f(X_1,\ldots, X_n)$. Obviously no previous knowledge about how is $f$ can be assumed and pretending a linear relationship has no sense,
but after admitting certain error in the model the assumption is valid.
\newline
\begin{center}
$y=\beta_1X_1+\beta_2X_2+\cdots+\beta_pX_p+\mathbf{\varepsilon}$
\end{center}

We are not interested in the credibility of the model nor its exactness, we only are interested in some sort of ordering of the coefficients $\beta$ so as we can assign a prevalence of the player $X_i$ on determining the result of the game.

Matricially $y=\mathsf{X}\mathsf{\beta}$ in a sort of an overdetermined and non compatible linear system. But, from a customary least square technique, we can consider the corresponding normal equations $\mathsf{X}^T y=\mathsf{X}^T\mathsf{X}\mathsf{\beta}$ and, therefore $\mathsf{\beta}= \left(\mathsf{X}^T\mathsf{X}\right)^{-1}\,\mathsf{X}^t\mathsf{y}$ 
\begin{enumerate}
    \item Normalizing data:
    Every player, i.e. every column $X_j$, is included in exactly a half of the number of total coalitions, hence $\overline{X_j}=\dfrac{2^{n-1}}{2^{n}}=\dfrac{1}2$. Moreover, since all coalitions of the game appears, there is the same amount of $0$ and $1$ in $X_j$ hence the corresponding variance is $\sigma_{X_j}^2=\dfrac{2^{n-1}(\dfrac{1}{2})^2+2^{n-1}(-\dfrac{1}{2})^2}{2^{n}}=\dfrac{1}{4}$. This normalization only implies that zeroes become negative ones. Hence, let $X^*$ denote the normalized matrix,
    \begin{equation*}
            x_{i,j}^*= \begin{cases}
             1 & \text{ if }x_{i,j}=1 \\
            -1 & \text{ if }x_{i,j}=0 
            \end{cases}
    \end{equation*}
    To normalize $Y$, note that the average $\overline{Y}=\frac{|\mathcal{W|}}{2^n}$, and the standard deviation corresponds to $\sigma_Y=\frac{\sqrt{|\mathcal{W}|(2^n-|\mathcal{W}|)}}{\sqrt{2^n}}$, thus
\begin{equation*}
    y_i^*=\frac{y_i-\frac{|\mathcal{W}|}{2^n}}{\frac{\sqrt{|\mathcal{W}|(2^n-|\mathcal{W}|)}}{\sqrt{2^n}}}
\end{equation*}
    \item First: $s=\mathsf{X}^Ty$

    We compute the standard regression coefficients \footnote{Note this coefficients would be related with the covariance functions if the variables were not standardized.}
    \begin{equation*}    s_j=\frac{X_{j}^{*T} y^*}{\sqrt{X_{j}^{*T}X_{j}^*}}=\frac{\sum\limits_{i=1}^{2^n}x_{ij}^{*}\cdot y_i^*}{\sqrt{X_{j}^{*T}X_{j}^*}}
\end{equation*}    
Note the denominator $\sqrt{X_j^{*T}X_j^*}$ is $\sqrt{2^n}$ by Lemma~\ref{lemma:ortho}. Then
    \begin{equation*}    s_j=\frac{1}{\sqrt{2^n}}\sum\limits_{i=1}^{2^n}\frac{x_{i,j}^{*} (y_i-\frac{|\mathcal{W}|}{n})}{\sqrt{\frac{|\mathcal{W}|(2^n-|\mathcal{W}|)}{2^n}}}=\frac{\sum_{i=1}^{2^n}(x_{i,j}^{*} y_i-x_{i,j}^{*}\frac{|\mathcal{W}|}{n})}{\sqrt{|\mathcal{W}|(2^n-|\mathcal{W}|)}}
\end{equation*}
The second term in the sum is zero since $x_{i,j}^*$ takes values $1$ for the half of the $i$ and $-1$ for the other half. So,

  \begin{equation*}    s_j=\frac{\sum_{i=1}^{2^n}x_{i,j}^{*} y_i}{\sqrt{|\mathcal{W}|(2^n-|\mathcal{W}|)}}=\frac{\sum_{i=1}^{2^n}(x_{i,j}-\frac{1}{2}) y_i}{\frac{1}{2}\sqrt{|\mathcal{W}|(2^n-|\mathcal{W}|)}}=\frac{\sum_{i=1}^{2^n}x_{i,j}y_i-\frac{1}{2} y_i}{\frac{1}{2}\sqrt{|\mathcal{W}|(2^n-|\mathcal{W}|)}}=\frac{|\mathcal{W}_j|-\frac{1}{2} |\mathcal{W}|}{\frac{1}{2}\sqrt{|\mathcal{W}|(2^n-|\mathcal{W}|)}}.
\end{equation*}
(només es sumen els xij quan el player j forma part d'una coalició guanyadora i)

Thus we can represent the vector of standard regression coefficients by
\begin{equation*}
    s^T=\frac{2}{\sqrt{|\mathcal{W}|(2^n-|\mathcal{W}|)}}\left(|\mathcal{W}_1|-\frac{1}{2} |\mathcal{W}|,\dots,|\mathcal{W}_n|-\frac{1}{2} |\mathcal{W}|\right).
\end{equation*}
    \item Second: $\sqrt{2^n}\,\mathsf{I}=\mathsf{X}^T\mathsf{X}$
    
    Les columnes de les dades normalitzats són perpendiculars entre elles i totes tenen el mateix mòdul (ja s'ha observat abans i al Lema 4.2). Així $\sqrt{2^n}\,\mathsf{I}=\mathsf{X}^T\mathsf{X}$ i $\left(\mathsf{X}^T\mathsf{X}\right)^{-1}=\dfrac{1}{\sqrt{2^n}}\,\mathsf{I}$
\end{enumerate}

After all we obtain $\mathsf{\beta}=
\left(\mathsf{X}^T\mathsf{X}\right)^{-1}\,
\mathsf{X}^T\mathsf{y}=
\dfrac{1}{\sqrt{2^n}}\,\mathsf{I}\,s=
\dfrac{1}{\sqrt{2^n}}\,s$

\end{enumerate}

\end{proof}

Note the construction of $M(v)$ involving all coalition instead of all minimal winning coalition was chosen to get an orthonormal basis.

Note that all the interpretations provided by PCA and SVD are valid and useful even in this case, where the process to transform the original components into an orthonormal once consists of applying the identity since the original ones are already orthonormal.

\section{Conclusion}

The SVD procedure applies whenever the input dataset $X$ and its target $Y$ lies in $\mathbb{R}$. Hence, there is room to try to extend this index into a wider class of games. For example, allowing the target to be in $\mathbb{R}$ we can extend the index into cooperative games, while allowing the dataset $X$ to be a natural number we can extend the index into the class of $(j,2)$ simple games with ordered inputs (\cite{frzw03}).

In future work, we plan to develop additional indices based on PCA and SVD. We will compare their effectiveness and assess their performance using approximation techniques when exact computation is infeasible.

\bibliographystyle{apacite}

\bibliography{references}

\end{document}